\documentclass[11pt,envcountsame]{article}

\usepackage[a4paper,margin=1in]{geometry}
\usepackage{xspace}
\usepackage{comment}
\usepackage{amsmath}
\usepackage{amsfonts}
\usepackage{amssymb}
\usepackage{amsthm}
\usepackage{thmtools,thm-restate}
\usepackage[textsize=small]{todonotes}
\usepackage{tikz}
\usetikzlibrary{arrows,automata}
\usepackage{stmaryrd}
\usepackage{mathdots}
\usepackage{algorithm}
\usepackage[noend]{algpseudocode}
\usepackage{thmtools}

\theoremstyle{definition}\newtheorem{theorem}{Theorem}
\theoremstyle{definition}\newtheorem{corollary}[theorem]{Corollary}
\theoremstyle{definition}\newtheorem{example}[theorem]{Example}
\theoremstyle{definition}\newtheorem{lemma}[theorem]{Lemma}
\theoremstyle{definition}
\theoremstyle{definition}
\theoremstyle{definition}
\theoremstyle{definition}
\theoremstyle{definition}\newtheorem{conjecture}[theorem]{Conjecture}

\newcommand{\ignore}[1]{}

\newcommand{\F}{\mathcal{F}}

\newcommand{\A}{\mathcal{A}}

\newcommand{\N}{\mathbb{N}}
\newcommand{\Z}{\mathbb{Z}}

\newcommand{\set}[1]{\{#1\}}

\newcommand{\reach}{\textsc{Reach}}

\newcommand{\trans}[1]{\stackrel{#1}{\longrightarrow}}

\newcommand{\trg}{\textup{trg}}
\newcommand{\src}{\textup{src}}

\newcommand{\conf}{\textup{Conf}}

\newcommand{\inp}{\textup{in}}
\newcommand{\out}{\textup{out}}
\newcommand{\size}{\textup{size}}
\newcommand{\flush}{\textbf{\textup{flush}}}
\newcommand{\ztest}{\textbf{\textup{zero-test}}}
\newcommand{\mult}{\textbf{\textup{multiply}}}
\newcommand{\reaches}{\longrightarrow}

\newcommand{\ampl}{\textup{\textbf{ampl}}\xspace}

\newcommand{\eff}{\textup{eff}}

\newcommand{\nl}{\textup{NL}\xspace}

\newcommand{\np}{\textup{NP}\xspace}
\newcommand{\pspace}{\textup{PSpace}\xspace}
\newcommand{\exptime}{\textup{ExpTime}\xspace}
\newcommand{\expspace}{\textup{ExpSpace}\xspace}

\newcommand{\tower}{\textup{Tower}\xspace}
\newcommand{\ackermann}{\textup{Ackermann}\xspace}

\newcommand{\add}[2]{$\coreadd{\vr{#1}}{#2}$}
\newcommand{\sub}[2]{$\coresub{\vr{#1}}{#2}$}
\newcommand{\coreadd}[2]{#1 \,\, +\!\!= \, #2}
\newcommand{\coresub}[2]{#1 \,\, -\!\!= \, #2}
\newcommand{\vr}[1]{#1}

\title{Lower Bounds for the Reachability Problem in Fixed Dimensional VASSes}

\author{Wojciech Czerwi\'nski \\ University of Warsaw \\ wczerwin@mimuw.edu.pl\footnote{Supported by the ERC grant INFSYS, agreement no. 950398.} \and
Łukasz Orlikowski \\ University of Warsaw \\ lo418363@students.mimuw.edu.pl\footnote{Supported by the Ministry of Science and Higher Education project Szkoła Orłów, project number 500-D110-06-0465160.}}

\begin{document}

\maketitle

\begin{abstract}
We study the complexity of the reachability problem for Vector Addition Systems with States (VASSes)
in fixed dimensions. We provide four lower bounds improving the currently known state-of-the-art:
1) \np-hardness for unary flat $4$-VASSes (VASSes in dimension 4),
2) \pspace-hardness for unary $5$-VASSes,
3) \expspace-hardness for binary $6$-VASSes
and 4) \tower-hardness for unary $8$-VASSes.
\end{abstract}

\section{Introduction}
Vector Addition Systems (VASes) together with essentially equivalent Petri nets and Vector Addition Systems with States (VASSes)
are fundamental models of computation with many application in practice and theory. Central algorithmic problem concerning
VASSes is the reachability problem asking whether in a given VASS there exists a run from one given configuration to another.
Long research history of this problem dates back to 70-ties when Lipton
has proven \expspace-hardness of the reachability problem~\cite{Lipton76}.
Decidability of the problem was shown a few years later by Mayr in~\cite{DBLP:conf/stoc/Mayr81},
where he presented a very involved algorithm.
After a few decades of research recently the complexity of the problem was settled to be \ackermann-complete.
The upper bound was shown by Leroux and Schmitz in~\cite{DBLP:conf/lics/LerouxS19} three years ago.
Last year \ackermann-hardness was independently proven by Leroux~\cite{DBLP:journals/corr/abs-2104-12695}
and by Czerwi\'nski and Orlikowski~\cite{DBLP:journals/corr/abs-2104-13866}.

Despite settling the computational complexity of the reachability problem in VASSes a lot of questions about VASSes
remain to be solved. Even the reachability problem is not fully understood and the most clear evidence for that is
the existence of big complexity gaps for the problem in small fixed dimensions. The prominent example
here is the dimension three with complexity gap between \pspace-hardness
(inherited from dimension two~\cite{DBLP:conf/lics/BlondinFGHM15})
and super-\tower (concretely speaking the $\F_7$, namely the $7$-th level of the Grzegorczyk hierarchy~\cite{DBLP:conf/lics/LerouxS19}).
The reachability problem was already extensively studied for fixed dimensions. For dimension one (i.e. for $1$-VASSes)
for binary encoding of numbers occurring in transitions it was shown to be \np-complete in~\cite{DBLP:conf/concur/HaaseKOW09}.
For unary encoded $1$-VASSes it is easy to see that the reachability problem is \nl-complete.
For $2$-VASSes the problem is known to be \pspace-complete in the case of binary encoding~\cite{DBLP:conf/lics/BlondinFGHM15}
and moreover \nl-complete in the case of unary encoding~\cite{DBLP:conf/lics/EnglertLT16}.
However, beyond dimension two the situation is much less clear.

In~\cite{DBLP:conf/concur/Czerwinski0LLM20} several cases of the reachability problem for fixed dimensional
VASSes were considered. In particular a subclass of flat VASSes was investigated, namely VASSes without nested loops
in the state structure.
This class was introduced in~\cite{DBLP:conf/concur/LerouxS04} and has a bunch of nice properties.
In particular the reachability relation is semilinear and the reachability problem can be easily shown to be in \np,
even in the case of binary encoding.
In~\cite{DBLP:conf/concur/Czerwinski0LLM20} it was shown that the reachability problem is \np-hard already for
a fixed dimension and unary encoding, namely for unary $7$-VASSes,
but the status of the problem for lower dimensions remained unsettled.

The first \expspace-hardness result for fixed dimension follows from~\cite{DBLP:conf/stoc/CzerwinskiLLLM19}, where it was shown
that the problem is $h$-\expspace-hard for unary $(h+13)$-VASSes, thus \expspace-hard for unary $14$-VASSes.
Recent \ackermann-hardness results delivered also \tower-hardness results in fixed dimensions.
Notice that \tower-hardness for binary $d$-VASSes implies \tower-hardness for unary $d$-VASSes
as the \tower complexity class is closed under exponential blowup of running time.
Thus we may not emphasise encoding when talking about \tower-hardness.
The dimension in which the problem is \tower-hard was step by step decreased from $21$ in the initial version of~\cite{DBLP:journals/corr/abs-2104-12695} and $18$ in~\cite{DBLP:journals/corr/abs-2104-13866} through dimension $17$ in third version of~\cite{DBLP:journals/corr/abs-2104-12695}, $11$ in the recent Lasota's work~\cite{DBLP:journals/corr/abs-2105-08551} to
a currently best value of $10$ in last version of~\cite{DBLP:journals/corr/abs-2104-12695}.
We further decrease the dimension and show that the reachability problem is already \tower-hard for $8$-VASSes.

\paragraph*{Our contribution}
We believe it is important to pursue the search for exact complexities for fixed dimensional VASSes.
First of all low dimensional VASSes are very natural computation models and currently known techniques
used to provide hardness results are very likely not to work in some small dimensions. Secondly, it is easier
to invent a sophisticated technique working in a simpler setting. Therefore it is quite
possible that the search for exact complexity bounds for the reachability problem in low dimensions will result
in finding new techniques useful in much broader generality. Thirdly, despite very high pessimistic complexity
of the reachability problem it still can be solved in practise in many cases~\cite{DBLP:conf/tacas/DixonL20}.
Therefore it is not only a theoretical, but may be also of practical interest to understand for which VASS subclasses
the reachability problem have relatively low complexity and avoiding which obstacles may lead to efficient algorithms.
One obvious way to pursue this idea is to understand better low dimensional VASSes.

Our main results are the four lower bound theorems, which improve the previously mentioned lower bounds.
Additionally we introduce a novel technique of proving lower bounds inspired by multiplication triples
technique introduced in~\cite{DBLP:conf/stoc/CzerwinskiLLLM19} and used also
in~\cite{DBLP:journals/corr/abs-2104-13866,DBLP:journals/corr/abs-2105-08551}.
We call it the quadratic pairs technique and use to decrease the dimension of VASSes
in certain hardness results. Concretely speaking we apply this approach to prove Theorems~\ref{thm:pspace}~and~\ref{thm:expspace}.

Beside that our main conceptual contribution is to compose already known techniques in a subtle way
in order to get lower bounds, which are 1) substantially stronger than currently known,
and 2) shown by some not very involved constructions.
We would like to emphasise that our constructions are rather simple,
but we see it as an advantage rather than a disadvantage.

As a first contribution we provide a simple construction which decreases the dimension
in which the reachability problem is \np-hard for unary, flat VASSes,
namely we decrease the dimension from $7$ in~\cite{DBLP:conf/concur/Czerwinski0LLM20} to a dimension $4$.

\begin{theorem}\label{thm:np}
The reachability problem for unary, flat $4$-VASSes is \np-hard.
\end{theorem}

We need only one dimension more to show \pspace-hardness for unary (not necessarily flat though) VASSes.

\begin{theorem}\label{thm:pspace}
The reachability problem for unary $5$-VASSes is \pspace-hard.
\end{theorem}

Next we lower the dimension for which \expspace-hardness is known from $10$~\cite{DBLP:journals/corr/abs-2104-12695} to $6$.

\begin{theorem}\label{thm:expspace}
The reachability problem for binary $6$-VASSes is \expspace-hard.
\end{theorem}

Notice that Theorem~\ref{thm:expspace} clearly shows also \pspace-hardness for unary $6$-VASSes, but for 
\pspace-hardness we can eliminate one dimension in the proof of Theorem~\ref{thm:pspace}.

We also show that only two dimensions more than needed for \expspace-hardness is enough to get \tower-hardness.

\begin{theorem}\label{thm:tower}
The reachability problem for unary $8$-VASSes is \tower-hard.
\end{theorem}

In order to prove our results we crucially exploit two known techniques designed to force counters of VASSes
to be equal to zero at some particular configurations along the run, namely simulate zero-tests on some counters.
The first technique is based on triples of the form $(B, C, BC)$ and was introduced in~\cite{DBLP:conf/stoc/CzerwinskiLLLM19} in order
to simulate $C / 2$ zero-tests for counters bounded by value $B$. This idea was later improved
in~\cite{DBLP:journals/corr/abs-2105-08551}~and~in~\cite{DBLP:journals/corr/abs-2104-13866} to handle many counters by just one triple.
Based on this technique we design our novel quadratic pair technique.
The second technique was introduced in~\cite{DBLP:journals/corr/abs-2104-13866} and uses a single controlling-counter
in order to perform a linear number of zero-tests. It turns out that none of these two tools dominate the other one, they
are useful in different situations.

\paragraph*{Organisation of the paper}
In Section~\ref{sec:prelim} we introduce preliminary notions and recall necessary facts about the above mentioned two techniques of zero-testing.
In Section~\ref{sec:prelim} we also introduce the quadratic pair technique and prove related facts about counter automata.
Then in Section~\ref{sec:overview} we briefly describe ideas beyond our proofs, in some cases it might be even sufficient to read
this section in order to understand in-depth our arguments.
In Sections~\ref{sec:np},~\ref{sec:pspace},~\ref{sec:expspace}~and~\ref{sec:tower}
we prove in detail Theorems~\ref{thm:np}, \ref{thm:pspace}, \ref{thm:expspace} and \ref{thm:tower}, respectively.
Finally in Section~\ref{sec:future} we comment about the limitations of our techniques and mention possible future research directions.

\section{Preliminaries}\label{sec:prelim}

\paragraph*{Basic notions}
For $a, b \in \N$ we write $[a,b]$ to denote the set $\{a, a+1, \ldots, b-1, b\}$.
For a vector $v \in \N^d$ and $i \in [1,d]$ we write $v[i]$ to denote the $i$-th coordinate
of vector $v$. By $0^d$ we denote vector $v \in \N^d$ with all coordinates equal to zero.

\paragraph*{Vector Addition Systems with States}
A $d$-dimensional Vector Addition System with States ($d$-VASS) consists of a finite
set of \emph{states} $Q$ and a finite set of transitions $T \subseteq Q \times \Z^d \times Q$.
\emph{Configuration} of a $d$-VASS is a pair $(q, v) \in Q \times \N^d$, we often write it $q(v)$ instead of $(q, v)$.
For a configuration $c = q(v)$ and $i \in [1,d]$ we denote by $c[i]$ value $v[i]$.
The set of all the configurations is denoted $\conf = Q \times \N^d$.
Transition $(p, u, q)$ can be \emph{fired} in a configuration $r(v)$ if $p = r$ and $u+v \in \N^d$.
We write then $p(v) \trans{(p,u,q)} q(u+v)$. The \emph{effect} of a transition $(p, u, q)$ is vector $u$,
we write $\eff((p, u, q)) = u$.
A sequence
$\rho = (c_1, t_1, c'_1), (c_2, t_2, c'_2), \ldots, (c_n, t_n, c'_n) \in \conf \times T \times \conf$
is a \emph{run} of VASS $V = (Q, T)$ if for all $i \in [1,n]$ we have $c_i \trans{t_i} c'_i$
and for all $i \in [1,n-1]$ we have $c'_i = c_{i+1}$.
We naturally extend the notion of the \emph{effect} to runs, $\eff(\rho) = \eff(t_1) + \ldots + \eff(t_n)$.
Such a run $\rho$ is \emph{from} configuration $c_1$ \emph{to} configuration $c'_n$.
We write then $c_1 \trans{\rho} c'_n$ slightly overloading the notation or simply $c_1 \reaches c'_n$
if there is some $\rho$ such that $c_1 \trans{\rho} c'_n$.
By $\reach(\src, V) = \{c \mid \src \reaches c\}$ we denote the set of all the configurations reachable from configuration $\src$
and we call it the \emph{reachability set}.
We also write simply $\reach(\src)$ is VASS $V$ is clear from the context.

The following problem is the main focus of this paper, for different values of $d \in \N$.

\begin{quote}\label{qu:problem}
\textbf{Reachability problem for $d$-VASSes}
\begin{description}
  \item[Input] A $d$-VASS $V$ and two its configurations $\src, \trg$
  \item[Question] Does $\src \reaches \trg$ in $V$?
\end{description}
\end{quote}

The size of VASS $V$, denoted $\size(V)$, is the total number of bits needed to represent states and transitions of $V$.
A \emph{state-cycle} in a VASS $V$ is a cycle in the graph $(Q,E)$ with vertices being states of $V$ and edges
being defined as $(p,q) \in E$ if there is some transition $(p,u,q) \in T$.
We say that a VASS $V$ is \emph{flat} if for each state $q \in Q$ there is at most one state-cycle in $V$ which contains $q$.
In other words a VASS is flat if there are no nested cycles in its state structure.
If numbers in transitions of a VASS are encoded in unary then we call it a \emph{unary VASS}. Similarly a \emph{binary VASS}
is a VASS with transitions encoded in binary.

\paragraph*{Counter programs}
A very useful formalism to describe some VASSes are counter programs.
A \emph{counter program} is a sequence of instructions of the form either \add{x}{a} or \textup{loop} P,
where $P$ is another counter program. Such a counter program with $d$ counters
can be transformed in a natural way to a corresponding $d$-VASS.
Thus in the rest of the paper in many places we use terms VASS and counter programs
almost interchangeably.
A precise definition can be found in~\cite{DBLP:journals/corr/abs-2104-13866},
we recall here examples provided in~\cite{DBLP:journals/corr/abs-2104-13866}.

\begin{example}
The following counter program
\begin{algorithmic}[1]
\State \add{\vr{x}}{1}
\Loop
\State \sub{\vr{x}}{1} \quad \add{\vr{y}}{1}
\EndLoop
\Loop
\State \add{\vr{x}}{2} \quad \sub{\vr{y}}{1}
\EndLoop
\Loop
\State \sub{\vr{x}}{1} \quad \add{\vr{y}}{1}
\EndLoop
\Loop
\State \add{\vr{x}}{2} \quad \sub{\vr{y}}{1}
\EndLoop
\end{algorithmic}
represents the $2$-VASS presented below, state names are chosen arbitrary.

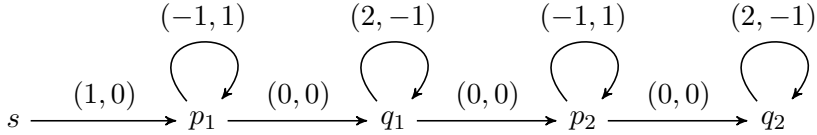
\begin{tikzpicture}[->,>=stealth',shorten >=1pt,auto,node distance=2.5cm,semithick]
\node (s) at (0,0) {$s$};
\node (p1) [right of=s] {$p_1$};
\node (q1) [right of=p1] {$q_1$};
\node (p2) [right of=q1] {$p_2$};
\node (q2) [right of=p2] {$q_2$};

\path[->]
(p1) edge[->, in=50, out=130, min distance=0.5cm,loop] node {$(-1,1)$}(p1)
(q1) edge[->, in=50, out=130, min distance=0.5cm,loop] node {$(2,-1)$}(q1)
(p2) edge[->, in=50, out=130, min distance=0.5cm,loop] node {$(-1,1)$}(p2)
(q2) edge[->, in=50, out=130, min distance=0.5cm,loop] node {$(2,-1)$}(q2)

(s) edge[->] node[above] {$(1,0)$} (p1)
(p1) edge[->] node[above] {$(0,0)$} (q1)
(q1) edge[->] node[above] {$(0,0)$} (p2)
(p2) edge[->] node[above] {$(0,0)$} (q2);
\end{tikzpicture}
\end{example}

\noindent
We often use macro \textbf{for} $i$ := $1$ \textbf{to} $n$ \textbf{do},
by which we represent just the counter program in which the body of the for-loop
is repeated $n$ times. 

\begin{example}\label{ex:2VASS}
The following counter program uses the macro \textbf{for}. For $n = 2$ it is equivalent to the above example.
\begin{algorithmic}[1]
\State \add{\vr{x}}{1}
\For {\, $i$ \, := \, $1$ \, \textbf{to } $n$}\label{l:for}
\Loop
\State \sub{\vr{x}}{1} \quad \add{\vr{y}}{1}
\EndLoop
\Loop
\State \add{\vr{x}}{2} \quad \sub{\vr{y}}{1}
\EndLoop
\EndFor
\end{algorithmic}

\noindent
The counter program represents the following $2$-VASS.

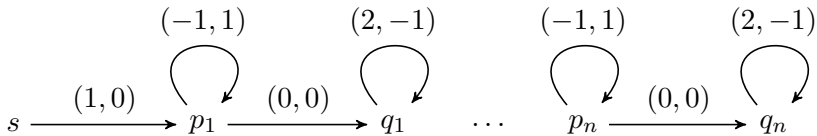
\begin{tikzpicture}[->,>=stealth',shorten >=1pt,auto,node distance=2.5cm,semithick]
\node (s)  at (0,0) {$s$};
\node (p1) [right of=s] {$p_1$};
\node (q1) [right of=p1] {$q_1$};
\node (pn) [right of=q1] {$p_n$};
\node (qn) [right of=pn] {$q_n$};

\node (dot) at (6.25, 0) {$\ldots$};

\path[->]
(p1) edge[->, in=50, out=130, min distance=0.5cm,loop] node {$(-1,1)$}(p1)
(q1) edge[->, in=50, out=130, min distance=0.5cm,loop] node {$(2,-1)$}(q1)
(pn) edge[->, in=50, out=130, min distance=0.5cm,loop] node {$(-1,1)$}(pn)
(qn) edge[->, in=50, out=130, min distance=0.5cm,loop] node {$(2,-1)$}(qn)

(s) edge[->] node[above] {$(1,0)$} (p1)
(p1) edge[->] node[above] {$(0,0)$} (q1)
(pn) edge[->] node[above] {$(0,0)$} (qn);
\end{tikzpicture}
\end{example}

Sometimes we add to counter programs an instruction $P_1 \  \textup{or} \  P_2$, where $P_1$ and $P_2$ are counter
programs. It is easy to see that such an instruction can be as well easily simulated by a nondeterministic choice in VASSes.

\paragraph*{Bounded counter automata}
A counter automaton is a VASS with special zero-test transitions, which can be fired only if a particular counter has value exactly zero.
It is a folklore that reachability problem for counter automata is undecidable in general. However restricted versions of the
problem are very natural problems complete for natural complexity classes. We say that a run of a counter automaton
is $B$-\emph{bounded} if \emph{the sum of all the counters} on that run has values smaller than $B$. 
Notice that here we use a bit unusual notion of boundedness: we demand the sum of all the counters to be bounded by $B$,
not every single counter by itself. This is however only a small technical change.
A run is \emph{accepting} if it starts in the distinguished initial state with all the counters equal to zero and
finishes in the distinguished accepting state also with all the counters equal to zero.
Consider the following problem:
\begin{quote}\label{qu:bounded-problem}
\textbf{The $f$-bounded reachability problem for $d$-counter automata}
\begin{description}
  \item[Input] A $d$-counter automaton $\A$, number $n \in \N$ given in unary
  \item[Question] Does $\A$ have an $f(n)$-bounded accepting run?
\end{description}
\end{quote}
The following theorem is a folklore. The proof can be found in~\cite{FischerMR68} (Theorem 3.1)
while in~\cite{DBLP:journals/toct/Schmitz16} (Section 4.1) 
it is argued that small modifications in the definition of the $\tower$ function do not change the class.

\begin{theorem}\label{thm:bounded-reachability}
The $f$-bounded reachability problem for three-counter automata is
\tower-complete for $f(n) = \tower(n)$ defined as $\tower(1) = 2$, $\tower(n+1) = 2^{\tower(n)}$ for any $n > 1$.
\end{theorem}

Theorem~\ref{thm:bounded-reachability} will be used in the \tower-hardness proof in Section~\ref{sec:tower}.
Actually in the case of $f = \tower$ even the problem for two-counter automaton is \tower-complete, but this
simplification is not needed in our construction.

For the \pspace-hardness and \expspace-hardness proofs in Sections~\ref{sec:pspace}~and~\ref{sec:expspace}, respectively
we need a more subtle problem. We call a counter automata to be \emph{$B$-bounded} if all its accepting runs are $B$-bounded.
Let us consider the following promise problem:
\begin{quote}\label{qu:bounded-problem}
\textbf{The reachability problem for $f$-bounded $d$-counter automata}
\begin{description}
  \item[Input] An $f(n)$-bounded $d$-counter automaton $\A$, number $n \in \N$ given in unary
  \item[Question] Does $\A$ have an accepting run?
\end{description}
\end{quote}
Notice that the assumption of $f$-boundedness makes the reachability problem for $f$-bounded counter automata
easier than the problem of $f$-bounded reachability for not necessarily bounded counter automata.
Indeed, if one can check $f$-bounded reachability for any counter automata then
in particular for $f$-bounded counter automata, for which it is equivalent to the reachability problem.
Thus the following theorem is harder to prove in our setting of the promise problem than in the more classical scenario.

\begin{theorem}\label{thm:bounded-counters}
The reachability problem for $f$-bounded $d$-counter automata is
\begin{enumerate}
  \item \pspace-hard for $f(n) = 2^n$ and $d = 2$
  \item \expspace-hard for $f(n) = 2^{2^n}$ and $d = 3$.
\end{enumerate}
\end{theorem}

The proof of Theorem~\ref{thm:bounded-counters} can be found in the Appendix.

In the next two paragraphs we present two different techniques, which can be used to simulate zero-tests in bounded counter automata
by a VASS without zero-tests.

\paragraph*{Controlling-counter technique}
Here we describe the technique of controlling-counter presented in~\cite{DBLP:journals/corr/abs-2104-13866}.
The essence of this technique is to add a new counter, called \emph{controlling-counter}, which is modified
in an appropriate way in the existing transitions and demanded to have value zero in both source and target configurations
of the run. This enforces that some other counters need to have zero values in particular configurations along the run.
If a counter is forced to be zero at some moment of the run we say that a \emph{zero-test} is performed on that counter
at this moment or it is \emph{zero-tested}.

Assume that configurations $c_1, \ldots, c_n$ are some of the configurations on run $\rho$ from configuration $\src$
to configuration $\trg$ and let
\[
c_0 \trans{\rho_1} c_1 \trans{\rho_2} \ldots \trans{\rho_n} c_n \trans{\rho_{n+1}} c_{n+1}.
\]
Let counter $x$ have value zero at both source $c_0$ and target $c_{n+1}$ of the run $\rho$ and let values
of counter $x$ in configurations $c_1, \ldots, c_n$ be $x_1, \ldots, x_n$ respectively, namely $c_i[x] = x_i$ for all $i \in [1,n]$.
Let $x'_i$ be the effect of run $\rho_i$ on counter $x$, namely $x'_1 = x_1$ and $x'_i = x_i - x_{i-1}$ for $i \in [2,n]$.
Clearly in order to assure $x_1 = x_2 = \ldots = x_n = 0$
it is enough to assure $x_1 + \ldots + x_n = 0$. Notice that for each $i \in [1,n]$ we have $x_i = x'_1 + x'_2 + \ldots + x'_i$.
Therefore
\[
x_1 + \ldots + x_n = n x'_1 + (n-1) x'_2 + \ldots + 2 x'_{n-1} + x'_n.
\]
Thus if there is a controlling-counter $y$ with the property that $c_0[y] = 0$ and for each $i \in [1,n]$ we
have $\eff(\rho_i)[y] = (n+1-i) \cdot \eff(\rho_i)[x]$ then we have that
\[
c_{n+1}[y] = n x'_1 + (n-1) x'_2 + \ldots + 2 x'_{n-1} + x'_n = x_1 + \ldots + x_n.
\]
Therefore $\trg[y] = 0$ implies that $c_i[x] = 0$ for all $i \in [1,n]$.

This idea can be extended to one counter controlling many counters. Here we recall Lemma 10 from~\cite{DBLP:journals/corr/abs-2104-13866} stating this generalised version, which will be used in our proofs.

\begin{lemma}\label{lem:zero-testing}
Let $\src \trans{\rho} \trg$ be a run of a $(d+1)$-VASS $V$
and let $\src = c_0, c_1, \ldots, c_{n-1}, c_n = \trg$ be some of the configurations on $\rho$.
Let $\rho_j$ for $j \in [1,n]$ be the parts of the run $\rho$ starting in $c_{j-1}$ and finishing in $c_j$, namely
\[
c_0 \trans{\rho_1} c_1 \trans{\rho_2} \ldots \trans{\rho_{n-1}} c_{n-1} \trans{\rho_n} c_n.
\]
Let $S_1, \ldots, S_d \subseteq [0,n]$ be the sets of indices of $c_j$, in which we want to zero-test counters numbered $1, \ldots, d$, respectively
and let $N_{j,i} = |\{k \geq j \mid k \in S_i\}|$ for $i \in [1,d], j \in [0,n]$ be the number of zero-tests, which we want to perform
on the $i$-th counter starting from configuration $c_j$ (in other words after the run $\rho_j$ for $j > 0$).
Then if:
\begin{enumerate}
  \item[(1)] $\src[d+1] = \sum_{i=1}^d N_{0,i} \cdot \src[i]$;
  \item[(2)] for each $j \in [1,n]$ we have $\eff(\rho_j, d+1) = \sum_{i=1}^d N_{j,i} \cdot \eff(\rho_j, i)$; and
  \item[(3)] $\trg[d+1] = 0$
\end{enumerate}
then for each $i \in [1,d]$ and for each $j \in S_i$ we have $c_j[i] = 0$.
\end{lemma}

\paragraph*{Multiplication triples technique}
The technique of multiplication triples was introduced in~\cite{DBLP:conf/stoc/CzerwinskiLLLM19}.
If values of counter $x$ along run $\rho$ are upper-bounded by $B$ then we say that $x$ is $B$-\emph{bounded} on $\rho$.
The essence of this idea is that a VASS starting with some three counters $b$, $c$ and $d$ having values $B$, $C$ and $BC$, respectively,
can perform $C/2$ zero-tests on a $B$-bounded counter.

Let us introduce a macro $\flush(x, y, z)$, which stands for a counter program:
\begin{algorithmic}[1]
\Loop
\State \sub{x}{1} \quad \add{y}{1} \quad \sub{z}{1}
\EndLoop
\end{algorithmic}
In other words $\flush(x, y, z)$ transfers value of counter $x$ to counter $y$ (but maybe not the whole value)
while keeping value $x+y$ constant and decreases counter $z$ by the transferred value.

Now it is easy to see how a zero-test on a $B$-bounded counter $x$ can be performed.
Assume as above that values of counters $(b, c, d)$ are $(B, C, BC)$
and initial value of $x$ is $0$. Then initially $x+b = B$ and as $x$ is $B$-bounded we can
keep this invariant along the run by decreasing $b$ when $x$ is increased and increasing $b$ when $x$ is decreased.
Then zero-test on $x$ is performed as follows:
\begin{algorithmic}[1]
\State \flush(b, x, d)
\State \flush(x, b, d)
\State \sub{c}{2}
\end{algorithmic}
In order to see that the above counter program indeed zero-tests $x$ notice that maximal decrease of $d$ during this program
is $2B$ and it is so only if $x = 0$, $b = B$ at the beginning of the program and both flushes were fully realised (so in particular
$x = 0$ also at the end of the program). Counter $c$ is decreased by $2$ in that program, therefore it can be fired at most $C / 2$
times, as the initial value of $c$ equals $C$. Thus in order to reach $d = 0$ at the end of the program each firing of zero-test must
result in decreasing $d$ by exactly $2B$. This in turn implies that zero-test can be indeed fired only if $x = 0$.

An extension of this technique to many counters zero-tested by the use of just one triple $(b, c, d)$ was introduced in~\cite{DBLP:journals/corr/abs-2104-13866} and elegantly described by Lasota in~\cite{DBLP:journals/corr/abs-2105-08551}.
We recall the argument here in order to be self-contained.
Assume now that we have $m$ counters $x_1, \ldots, x_m$ which all have value zero at the beginning of the counter program
and their sum $x_1 + \ldots + x_m$ is bounded by $B$ along the run. Then triple $(B, C, BC)$ allows for $C / 2$ zero-tests
on any of $x_1, \ldots, x_m$. We show how to perform zero-test on counter $x_1$, zero-testing other counters is very similar,
we comment about it in a moment.
\begin{algorithm}
\caption{}
\label{alg:zero-test}
\begin{algorithmic}[1]
\State \flush($x_2$, $x_1$, d)
\State \flush($x_3$, $x_2$, d)
\State \ldots
\State \flush($x_m$, $x_{m-1}$, d)
\State \flush(b, $x_m$, d)
\State \flush($x_m$, b, d)
\State \flush($x_{m-1}$, $x_m$, d)
\State \ldots
\State \flush($x_2$, $x_3$, d)
\State \flush($x_1$, $x_2$, d)
\State \sub{c}{2}
\end{algorithmic}
\end{algorithm}
Notice that in lines 1-5 we are flushing values from counters with bigger indices to counters with smaller indices
and in lines 6-10 we do the same process backwards.

The main idea is similar as above: we argue that $d$ can be decreased maximally by $2B$ by the zero-test program
and if it is decreased exactly by $2B$ then $x_1 = 0$ at the moment of zero-test and values of all the counters $x_i$ and $b$
are the same before and after the zero-test. Clearly the rest of the argument works as before, so it suffices to show the above property.
Let us denote for a moment counter $b$ by $x_{m+1}$, let $a_i$ be the value of $x_i$ at the beginning of the program
and $a'_i$ be the value of $x_i$ after $\flush$ in line 5. Clearly 
$a_1 + \ldots + a_m + a_{m+1} = a'_1 + \ldots + a'_m + a'_{m+1} = B$.
Notice that total decrease of $d$ in lines 1-5 is bounded by $a_2 + \ldots + a_{m+1}$ and
total decrease of $d$ in lines 6-10 is bounded by $a'_1 + \ldots + a'_m$. Therefore
total decrease is bounded by:
$a_2 + \ldots + a_{m+1} + a'_1 + \ldots + a'_m = B - a_1 + B - a'_{m+1} = 2B - (a_1 + a'_{m+1})$.
Thus clearly total decrease of $z$ is at most $2B$ and it equals $2B$ if: 1) $a_1 = a'_{m+1} = 0$;
and 2) all the flushes are fully realised. One can easily see that if all the flushes are fully realised
and $a_1 = 0$ then final values of $x_i$ are the same as the original ones, so the zero-test indeed works are required.
In order to zero-test counter different than $x_1$, say $x_i$ we perform the same procedure, but we apply flushes in different
order such that $x_i$ takes the place of counter $x_1$.

Recall now that an accepting run of a bounded counter automaton is from the distinguished initial state with all counters
having zero values to the distinguished final state with all counters having zero values.
Thus we can summarise the reasoning described above in the following lemma.
\begin{lemma}\label{lem:triples}
For each $d$-counter automaton $\A$
which on its $B$-bounded accepting run fires at most $C$ zero-tests
one can construct a unary $(d+3)$-VASS $V$ with two distinguished states $q_I, q_F$ such that:
$\A$ has an accepting run if and only if
there is a run from $q_I(B, 2C, 2BC, 0^d)$ to $q_F(B, 0^{d+2})$ in $V$.
\end{lemma}

\paragraph*{Quadratic pairs technique}
We emphasise here that in order to apply this technique we need to work with $B$-bounded counter automata,
rather than with $B$-bounded runs of not necessarily bounded counter automata, in contrast to the multiplication triple technique.
This is because in the multiplication triple technique the counters are checked to be bounded, while in the quadratic pairs technique
the counters are not checked to be bounded, we need to know in advance that they are $B$-bounded.
The essence of this idea is that a VASS starting with some two counters $b$ and $c$ having values $2B$ and $4B^2$, respectively,
can perform $B$ zero-tests on a $B$-bounded counter.

We first illustrate this technique for one counter $x$ and then show how to easily
generalise it to more counters. Assume that values of $(b, c)$ are $(B, B^2)$ and initial value of $x$ is $0$. Then initially $x + b = B$,
thus we have that $(x+b)^2 = B^2 = c$. The idea of the technique is that we keep the invariant $(x+b)^2 = c$
along the run as long as all the performed zero-tests are correct.
If at some moment an incorrect zero-test is fired then $(x+b)^2 < c$ and this inequality holds till the end of the run
implying in particular that $0 < c$. Thus checking $c = 0$ at the end of the run shows that all the performed zero-tests were correct.

The zero-test on $x$ is performed as follows:
\begin{algorithmic}[1]
\State \flush(b, x, c)
\State \flush(x, b, c)
\State \sub{b}{1} \quad \add{c}{1}
\end{algorithmic}
If initially $c = B^2$ and $x + b = B$ then after lines 1-2 still $x + b = B$ and $c \geq B^2 - 2B$ where the equality holds
if and only if both flushes were fully realised.
Thus after line 3 we have $x + b = B-1$ and $c \leq B^2 - 2B + 1 = (B-1)^2 = (x+b)^2$ and the equality holds iff both flushes were fully
realised, so in particular the zero-test was correct.

Thus after $\ell \leq B/2$ zero-tests performed on $x$ we have $b + x = B - \ell \geq B - B/2 = B/2$.
As we know that $x$ is $B/2$-bounded then applying at most $B/2$ zero-tests on $x$ is possible (as $x <= B - B/2 = B/2$).
We know after $\ell \leq B/2$ zero-tests that they were correct if $(x + b)^2 = c$. But after $\ell$ zero-tests $x + b = B - \ell$, so it is
not immediately how to check whether equality $(x+b)^2 = c$ holds. We check it by performing artificial zero-tests at the end of the run.
Namely as the last step we allow for arbitrary decrease of counter $x$ and arbitrarily many artificial zero-tests.
After each artificial zero-test the following invariant is kept: all the zero-tests are correct only if $(x+b)^2 = c$, otherwise
$(x+b)^2 < c$; and additionally if all the zero-tests are correct then it is possible to have $(x+b)^2 = c$.
Thus in order to check whether all the zero-tests were correct it is enough to check at the very end whether $c = 0$, similarly
as in the multiplication triple technique.

One can easily observe that extending this technique to many counters $x_1, \ldots, x_m$ which are $B/2$-bounded (recall that this means that its
sum is bounded by $B/2$) is straightforward. The only modification is the implementation of zero-tests which decrease the counter $c$
exactly by two times of the current value of $b + x_1 + \ldots + x_\ell$.
This is realised exactly as in the multiplication triple technique, namely as presented in Algorithm~\ref{alg:zero-test}.
The above reasoning can be summarised in the following lemma.

\begin{lemma}\label{lem:pairs}
For each $B$-bounded $d$-counter automaton $\A$ which on its accepting run fires at most $B$ zero-tests
one can construct in polynomial time a unary $(d+2)$-VASS $V_\A$
with two distinguished states $q_I, q_F$ such that
the following are equivalent:
\begin{enumerate}
  \item $\A$ has an accepting run
  \item there is a run from $q_I(2B, 4B^2, 0^d)$ to $q_F(0^{d+2})$ in $V_\A$.
\end{enumerate}
\end{lemma}

Using Lemma~\ref{lem:pairs} and Theorem~\ref{thm:bounded-counters} one can pretty easily get some hardness results
for VASS reachability problems, namely Corollaries~\ref{cor:pspace}~and~\ref{cor:expspace}.
As an intermediate tool for these results we formulate the following lemma.

\begin{lemma}\label{lem:simulation}
For each $B$-bounded $d$-counter automaton $\A$ with $s$ states
one can construct in polynomial time a unary $(d+2)$-VASS $V_\A$
with two distinguished states $q_I, q_F$ such that
the following are equivalent:
\begin{enumerate}
  \item $\A$ has an accepting run
  \item there is a run from $q_I(\bar{B}, \bar{B}^2, 0^d)$ to $q_F(0^{d+2})$ in $V_\A$ where $\bar{B} = 2sd \cdot B^{d-1}$.
\end{enumerate}
\end{lemma}

\begin{proof}
First observe that if there is an accepting run of the counter automaton $\A$ then there is also an accepting run with no repeating configuration.
Notice that the number of zero-tests in a run with no repeating configuration is bounded by the total number of $B$-bounded
configurations in $\A$ with at least one counter equal to zero. The number of such configurations with zero counter value can be bounded
by $s \cdot d \cdot B^{d-1}$. Indeed, there are at most $s$ choices of the state of the configuration, at most $d$ choices of the
counter, which equals to zero and at most $B^{d-1}$ choices for values of the other counters (some configurations are counted many times,
but this only strengthens the bound). Thus if there is an accepting run then there is an accepting run with at most $sdB^{d-1} = \bar{B} / 2$
zero-tests performed for $\bar{B}$ defined in the lemma statement. Notice now that if $\A$ is $B$-bounded then it is also $\bar{B}$-bounded
as $B \leq \bar{B}$. Then using Lemma~\ref{lem:pairs} applied to $\bar{B}$-bounded $d$-counter automaton $\A$ finishes the proof.
\end{proof}

The following corollaries are immediate consequences of Theorem~\ref{thm:bounded-counters} and Lemma~\ref{lem:simulation}.

\begin{corollary}\label{cor:pspace}
Given $n, s \in \N$ and a unary $4$-VASS $V$ with distinguished states $q_I, q_F$ it is $\pspace$-hard
to decide whether there is a run from $q_I(4s \cdot 2^n, 16s^2 \cdot 4^n, 0, 0)$ to $q_F(0^4)$.
\end{corollary}

\begin{corollary}\label{cor:expspace}
Given $n, s \in \N$ and a unary $5$-VASS $V$ with distinguished states $q_I, q_F$ it is $\expspace$-hard
to decide whether there is a run from $q_I(6s \cdot 4^{2^n}, 36s^2 \cdot 16^{2^n}, 0^3)$ to $q_F(0^5)$.
\end{corollary}

\section{Overview}\label{sec:overview}
Here we provide short sketches of the proofs of Theorems \ref{thm:np}, \ref{thm:pspace},
\ref{thm:expspace}~and~\ref{thm:tower}.
In the following sections we prove these theorems in detail.
Let us emphasise which techniques are used in which proofs.
Let us denote the controlling-counter technique by (CC), the multiplication triple technique by (MT)
and the quadratic pairs technique by (QP). Then to prove Theorem~\ref{thm:np} we use (CC),
to prove Theorem~\ref{thm:pspace} we use (CC) and (QP), to prove Theorem~\ref{thm:expspace} we use (MT) and (QP)
and to prove Theorem~\ref{thm:tower} we use (CC) and (MT).

\paragraph*{Proof of Theorem~\ref{thm:np}}
This is the easiest proof out of the four presented ones.
We reduce from the \textsc{Subset Sum} problem asking whether there is a subset of the set $\{s_1, \ldots, s_n\} \subseteq \N$
which sums up to a given number $s \in \N$. The main challenge is that numbers $s_i$ and $s$ in \textsc{Subset Sum} are
encoded in binary, while transitions in our $4$-VASS are encoded in unary. We use VASSes very similar to the one from
Example~\ref{ex:2VASS} in order to be able to obtain exponential counter values out of unary encoded numbers in VASS transitions.
If we add the third counter, which is a controlling-counter, we are able to construct a flat, unary $3$-VASS, which produces a number $s_i$
on a distinguished counter. Then we reduce the \textsc{Subset Sum} problem as follows: we have a distinguished counter called
the \emph{summing} counter, to which we first add value $s$ using a $3$-VASS (then altogether we have four counters).
Then for each $i \in [1,n]$ we construct a $3$-VASS, which produces number $s_i$ and then nondeterministically: either subtracts $s_i$
from the summing counter or does not touch the summing counter. After processing all the $3$-VASSes for $s_1, \ldots, s_n$
we check the summing counter to be zero: it is easy to observe that there exists a run reaching zero if and only if
the instance of \textsc{Subset Sum} is positive.

\paragraph*{Proof of Theorem~\ref{thm:pspace}}
By Corollary~\ref{cor:pspace} to show \pspace-hardness it is enough to design for given $s, n \in \N$ a $5$-VASS,
or in other words a five counter program of size polynomial in $s$ and $n$ which constructs on its first four counters $(x_1, x_2, x_3, x_4)$
values $(4s \cdot 2^n, 16s^2 \cdot 4^n, 0, 0)$ under the condition that $x_5 = 0$. Indeed, then checking whether
it reaches valuation $0^5$ at its end is \pspace-hard by Corollary~\ref{cor:pspace}.
We construct the pair $(4s \cdot 2^n, 16s^2 \cdot 4^n)$ on $(x_1, x_2)$ in the following way.
We start with $(x_1, x_2) = (4s, 16s^2)$ and then exactly $n$ times multiply $x_1$ by $2$ and $x_2$ by $4$.
The multiplications are realised as flushing $x_1$ or $x_2$ to $x_3$ and then flushing it back from $x_3$ to $x_1$ or $x_2$
simultaneously multiplying it by $2$ or $4$, respectively. We assume that multiplications are exact by forcing
appropriate counters $x_1$, $x_2$ and $x_3$ to be exactly zero after the flushes. This is realised by the use of controlling-counter
technique, the counter $x_5$ controls $x_1$, $x_2$ and $x_3$ thus if $x_5 = 0$ at the end of the run then all the multiplications were
indeed exact. Thus indeed after this phase the five counters have values $(4s \cdot 2^n, 16s^2 \cdot 4^n, 0, 0, x_5)$
under the condition that $x_5 = 0$.

\paragraph*{Proof of Theorem~\ref{thm:expspace}}
The idea is similar to the proof of Theorem~\ref{thm:pspace}.
By Corollary~\ref{cor:expspace} to show \expspace-hardness it is enough to design for given $s, n \in \N$ a $6$-VASS,
or in other words a six counter program of size polynomial in $s$ and $n$ which constructs on its first five counters $(x_1, x_2, x_3, x_4, x_5)$
values $(6s \cdot 4^{2^n}, 36s^2 \cdot 16^{2^n}, 0, 0, 0)$ under the condition that $x_6 = 0$. Indeed, then checking whether
it reaches valuation $0^6$ at its end is \expspace-hard by Corollary~\ref{cor:expspace}.
We construct the pair $(6s \cdot 4^{2^n}, 36s^2 \cdot 16^{2^n})$ on $(x_1, x_2)$ in the following way.
We start from setting $(x_1, x_2) = (6s, 36s^2)$ and then $2^n$ times we perform the following:
1) flush $x_1$ to $x_3$, 2) flush back $x_3$ to $x_1$ while multiplying by $4$, 3) flush $x_2$ to $x_3$,
4) flush back $x_3$ to $x_2$ while multiplying by $16$. After each flush be perform a zero-test to assure that the flush
was full. Additionally after these multiplications we perform a zero-test on $x_4$.
This time we cannot use the controlling-counter technique easily, as the number of zero-tests is equal to $4 \cdot 2^n + 1$,
which is super-linear. In order to simulate $4 \cdot 2^n + 1$ zero-tests (even on big counters) we use the multiplication triples technique.
We produce triple $(B, 8 \cdot 2^n + 2, B \cdot (8 \cdot 2^n+2))$ on counters $(x_4, x_5, x_6)$
for some big guessed value $B \in \N$ and use it to implement $4 \cdot 2^n + 1$ zero-tests
on $B$-bounded counters. Using this triple and checking that at the end of the run counter $x_6$ has value zero
guarantees that indeed all the flushes were full.
So after this phase we indeed have values $(6s \cdot 4^{2^n}, 36s^2 \cdot 16^{2^n}, 0, 0, 0)$ on the first five counters.

\paragraph*{Proof of Theorem~\ref{thm:tower}}
We reduce from the $\tower(n)$-bounded reachability problem for three-counter automata.
This construction uses both the multiplication triples technique and the controlling-counter technique
in an interplay. The aim is, similarly as in the proof of Theorem~\ref{thm:expspace},
to construct a triple of the form $(\tower(n), C, C \cdot \tower(n))$ for appropriately big $C$.
We first show that there exists a $7$-VASS, which is a $2^k$-amplifier (more precisely speaking an $f$-amplifier for $f(k) = 2^k$).
The notion of an amplifier was defined in~\cite{DBLP:journals/corr/abs-2104-13866}, we recall it in Section~\ref{sec:tower}.
Roughly speaking a $2^k$-amplifier from a triple $(B, C, BC)$ produces a triple $(2^B, C', C' \cdot 2^B)$
for some guessed value $C' \in \N$. In short words the construction of the amplifier works as follows: we start from a triple $(1, C', C')$
for guessed $C'$ and then using the triple $(B, C, BC)$ multiply exactly $B / 8$ times the first
and the third coordinate of the triple $(1, C', C')$ by exactly $2^8 = 256$.
After these multiplications we therefore get a triple $(2^{8 \cdot B/8}, C', C' \cdot 2^{8 \cdot B/8}) = (2^B, C', C' \cdot 2^B)$ as needed.
Using the trick from the previous paragraph we are able to achieve it by the use of just one additional counter and therefore the $2^k$-amplifier
has only seven counters. We use then the eighth counter as a controlling-counter: we compose the $2^k$-amplifier exactly $n$ times
and assure by the controlling-counter that the appropriate counters in the places of composition have value exactly zero,
which guarantees that composition works correctly.
As the number of compositions is linear this can be achieved by a single controlling-counter and thus the whole construction
uses only eight counters.  
\section{\np-hardness for 4-VASSes}\label{sec:np}
We reduce from the following problem:
\begin{quote}\label{qu:subset-sum}
\hskip -0.08cm
\textbf{\textsc{Subset Sum} problem}
\begin{description}
  \item[Input] Number $s_0 \in \N$, set of numbers $S = \{s_1, \ldots, s_n\} \subseteq \N$, all encoded in binary 
  \item[Question] Is there a subset of $S$ summing up exactly to $s_0$?
\end{description}
\end{quote}

For an instance of \textsc{Subset Sum} we 
design a four-counter program $P$ and show that there is a run of $P$ starting in $0^4$ and finishing in $0^4$ iff the instance is positive.
Our counter program has four counters: $x$ and $y$, which will be used to generate numbers $s_i$,
the summing counter $z$ and the controlling counter $c$.
The counter program $P$ consists of counter program $P_0$ and for each $i \in [1,n]$ counter programs $P_i$ and $P'_i$
in the following way:
\begin{algorithmic}[1]
\State $P_0$
\For {\, $i$ \, := \, $1$ \, \textbf{to } $n$}
\State $P_i \quad \textup{or} \quad P'_i$
\EndFor
\end{algorithmic}
The counter program $P_0$ will be constructed such that in every run reaching $0^4$ its effect on counter $z$ is exactly $s_0$.
On the other hand in such runs the effect of counter programs $P_i$ on $z$ for $i \in [1,n]$ will be exactly $-s_i$,
while $P'_i$ will have no effect on $z$.

Let us assume that $2^k$ is the smallest power of $2$ strictly bigger than all the numbers $s_0, s_1, \ldots, s_n$,
namely all $s_i$ can be encoded in $k$ bits.
For each $i \in [0,n]$ let $s_i = \langle b^i_{k-1} \cdots b^i_0 \rangle_2$ be the bit representation of $s_i$.
We show now how the counter program $P_0$ is constructed.
For simplicity we first do not contain the controlling counter $c$ in the program.
\begin{algorithmic}[1]
\State \add{x}{b^0_{k-1}}
\For {\, $j$ \, := \, $k-2$ \, \textbf{downto } $0$}
\Loop
\State \sub{x}{1} \quad \add{y}{1}
\EndLoop
\Loop
\State \add{x}{2} \quad \sub{y}{1}
\EndLoop
\State \add{x}{b^0_j}
\EndFor
\Loop
\State \sub{x}{1} \quad \add{z}{1}
\EndLoop
\end{algorithmic}
One can easily see that if loops in lines 3-4, 5-6 and 8-9 are fired maximal possible number of times
then the final values of $(x, y, z)$ are $(0, 0, s_0)$. Therefore to guarantee that $z = s_0$ after program $P_0$ it
is enough to assure that $x = 0$ each time line 4 is left, $y = 0$ each time line 6 is left
and $x = 0$ when line 9 is left. In order to achieve that we use the counter $c$.
However its behaviour depends as well on programs $P_i$ and $P'_i$, so we first present them,
also without the controlling counter. We first show counter program $P_i$ also without counter $c$.
\begin{algorithmic}[1]
\State \add{x}{b^i_{k-1}}
\For {\, $j$ \, := \, $k-2$ \, \textbf{downto } $0$}
\Loop
\State \sub{x}{1} \quad \add{y}{1}
\EndLoop
\Loop
\State \add{x}{2} \quad \sub{y}{1}
\EndLoop
\State \add{x}{b^i_j}
\EndFor
\Loop
\State \sub{x}{1} \quad \sub{z}{1}
\EndLoop
\end{algorithmic}
The only difference between $P_i$ for $i \geq 1$ and $P_0$ is that in $P_i$
in the loop in lines 8-9 the counter $z$ is decreased, while in $P_0$ it was increased.
One can easily observe that if $P_i$ for $i \geq 1$ starts with valuation $(x, y, z) = (0, 0, N)$ and
all the loops are iterated maximal number of times then it finishes with valuation $(x, y, z) = (0, 0, N-s_i)$.
Counter program $P'_i$ (also with counter $c$ ignored) is the same as $P_i$ with the only difference
that in line 9 counter $z$ is not decreased, but kept unchanged.
Intuitively the run in VASS choses to use $P_i$ if the number $s_i$ have to be taken into the sum and $P'_i$
if the number $s_i$ is not taken into the sum.
We can see now that counter programs $P_0$, $P_i$ and $P'_i$ have the promised properties under the condition
that counters $x$ and $y$ are zero in the appropriate places. In order to assure it we add the controlling counter $c$.
One can observe that $P_0$, $P_i$ and $P'_i$ differ only on the operation done to $z$ in line 9: in $P_0$ it is increase by $1$,
in $P'_i$ it is increased by $0$ and in $P_i$ it is increased by $-1$.
Therefore we write one parametrised program to represent all the three counter programs.
The presented counter program $\bar{P}(i, \textup{sign})$ satisifes $P_0 = \bar{P}(0, 1)$, $P_i = \bar{P}(i, -1)$
and $P'_i = \bar{P}(i, 0)$.
\begin{algorithmic}[1]
\State \add{x}{b^i_{k-1}} \quad \add{c}{b^i_{k-1} \cdot k(n-i+1)}
\For {\, $j$ \, := \, $k-2$ \, \textbf{downto } $0$}
\Loop
\State \sub{x}{1} \quad \add{y}{1} \quad \sub{c}{n+1-i}
\EndLoop
\Loop
\State \add{x}{2} \quad \sub{y}{1} \quad \add{c}{(k+1)(n-i)+(j+1)}
\EndLoop
\State \add{x}{b^i_j} \quad \add{c}{b^i_j \cdot (k(n-i)+(j+1))}
\EndFor
\Loop
\State \sub{x}{1} \quad \add{z}{\textup{sign}} \quad \sub{c}{k(n-i)+1}
\EndLoop
\end{algorithmic}
The only part in $\bar{P}(i, \textup{sign})$, which is nontrivial to understand are the effects of transitions on the controlling
counter $c$. Let us recall from Lemma~\ref{lem:zero-testing} that if counter $c$ controls counter $x$ then any increment
of \add{x}{a} should be matched by \add{c}{Na}, where $N$ is the number of zero-tests which are planned to be performed on
$x$ in the remaining part of the run. It is clear that there exist an appropriate changes of $c$,
which fulfil Lemma~\ref{lem:zero-testing} and they are not too big, so for an intuitive understanding of the program
one does not need the next paragraph. However in order to prove that we above counter program
indeed satisfies the needed conditions we need to meticulously inspect all the cases, which we do below.

In order to count the needed changes on $c$ we need to count how many times zero-tests are performed
on the controlled counters $x$ and $y$.
In each program $\bar{P}$ the counter $y$ is zero-tested $k-1$ times in the line 6, while counter $x$ is
zero-tested $k-1$ times in the line 6 and once in line 9, so altogether $k$ times.
Therefore in line 1 in program $\bar{P}(i, \textup{sign})$ counter $x$ is waiting for all the tests in programs $\bar{P}$
with first parameter being $i, i+1, \ldots, n$, altogether $n-i+1$ programs $\bar{P}$. Thus the number of zero-tests waiting
for $x$ is exactly $k(n-i+1)$ and each increment \add{x}{b_{k-1}^i} should be matched by
increment \add{c}{b^i_{k-1} \cdot k(n-i+1)}. Similarly in line 9 in program $\bar{P}(i, \textup{sign})$ counter $x$ is waiting for
$k(n-i)$ zero-tests in programs $\bar{P}$ with first parameter being $i+1, \ldots, n$ plus one last zero-test after
the loop in lines 8-9 in program $\bar{P}(i, \textup{sign})$. Thus \sub{x}{1} has to be matched with
\sub{c}{k(n-i) + 1}. A similar calculation shows that in line 7 increment \add{x}{b^i_j} has to be matched
with \add{c}{b^i_j \cdot (k(n-i)+(j+1))}. A bit more involved calculation is needed in case of lines 4 and 6, as there both counters
$x$ and $y$ are modified, so modification of $c$ has to reflect both changes.
In line 4 counter $x$ is waiting for $k(n-i)$ zero-tests in next programs, one zero-test in line 9
and $j+1$ zero-tests in line 4 in the further iterations of the for loops, so altogether for $k(n-i) + (j+2)$ zero-tests.
In the same line counter $y$ is waiting for $(k-1)(n-i) + (j+1)$ zero-tests, thus the total
change on $c$ is the $-k(n-i) - (j+2) + (k-1)(n-i) + (j+1) = - (n-i) - 1$. Similarly one can count that
in line 6 counter $x$ awaits for $k(n-i) + (j+1)$ zero-tests, while counter $y$ awaits for $(k-1)(n-i) + (j+1)$ zero-tests.
Therefore $c$ should be incremented by $2(k(n-i)+(j+1)) - (k-1)(n-i) - (j+1) = (k+1)(n-i) + (j+1)$.

One can easily observe that all the changes performed on $c$ are of polynomial size, therefore the reduction
from \textsc{Sumset Sum} is indeed performed in polynomial time. This finishes the proof of the Theorem~\ref{thm:np}.

\section{\pspace-hardness for 5-VASSes}\label{sec:pspace}
Due to Corollary~\ref{cor:pspace} in order to prove Theorem~\ref{thm:pspace} it is enough to show
the following lemma.

\begin{lemma}\label{lem:five-vass}
For each $s, n \in \N$ one can construct in polynomial time
a unary $5$-VASS of size polynomial in $s$ and $n$ with distinguished states $q_I, q_F$ such that
for each run from $q_I(0^5)$ to $q_F(x_1, x_2, x_3, x_4, 0)$
we have $(x_1, x_2, x_3, x_4) = (4s \cdot 2^n, 16s^2 \cdot 4^n, 0, 0)$.
\end{lemma}

Indeed, let $V_1$ be the $5$-VASS from Lemma~\ref{lem:five-vass} with distinguished states $q^1_I, q^1_F$.
Our aim is to reduce the problem from Corollary~\ref{cor:pspace} to the reachability problem in unary $5$-VASSes.
Let $V_2$ be a $4$-VASS from Corollary~\ref{cor:pspace} with distinguished states $q^2_I, q^2_F$ for
which we want to check whether there is a run from $q^2_I(4s \cdot 2^n, 16s^2 \cdot 4^n, 0, 0)$ to $q^2_F(0^4)$.
Let $V'_2$ be $V_2$ extended with the fifth coordinate in such a way that all the transitions have zero on this fifth coordinate.
We construct now a unary $5$-VASS $V$ with distinguished states $q^1_I, q^2_F$ which is a disjoined union
of $V_1$ and $V'_2$ with additional transition from $q^1_F$ to $q^2_I$ labelled by $0^5$. It is then immediate
to see that the following are equivalent:
\begin{itemize}
  \item there is a run from $q^2_I(4s \cdot 2^n, 16s^2 \cdot 4^n, 0, 0)$ to $q^2_F(0^4)$ in $V_2$
  \item there is a run from $q^1_I(0^5)$ to $q^2_F(0^5)$ in $V$,
\end{itemize}
which finishes the proof of Theorem~\ref{thm:pspace}.
Thus the rest of this section focuses on the proof of Lemma~\ref{lem:five-vass}.

\begin{proof}[Proof of Lemma~\ref{lem:five-vass}]
In the proof we prefer to use the terminology of counter programs instead of VASSes, but recall
that counter programs are just syntactic sugar to present VASSes in a human-readable way.
In our construction we actually do not use the counter $x_4$.
Our aim is to construct on $(x_1, x_2)$ values $(4s \cdot 2^n, 16s^2 \cdot 4^n)$.
We start with setting $(x_1, x_2)$ on $(4s, 16s^2)$.
Then we need to multiply $n$ times $x_1$ by $2$ and $x_2$ by $4$.
We realise it by flushing $x_1$ to $x_3$ and then flushing it back from $x_3$ to $x_1$ while
simultaneously multiplying by $2$, and similarly with $x_2$ but multiplying it by $4$.
In order to assure that all the multiplications
are exact we perform a zero-test after each flush, then we are sure that all the flushes are full.
Before explaining how we realise zero-tests we can already present how our counter program
works. Let us define the following macro $\mult(x, y, c)$ for two counters $x$ and $y$ and number $c \in \N$.
\begin{algorithmic}[1]
\Loop \quad \sub{x}{1} \quad \add{y}{1}
\EndLoop
\State \ztest($x$)
\Loop \quad \add{x}{c} \quad \sub{y}{1}
\EndLoop
\State \ztest($y$)
\end{algorithmic}
Using the macro $\mult(x, y, c)$ we can briefly describe our counter program as follows.
\begin{algorithmic}[1]
\State \add{x_1}{4s} \quad \add{x_2}{16s^2}
\For {\, $i$ \, := \, $1$ \, \textbf{to } $n$}\label{l:for}
\State $\mult(x_1, x_3, 2)$
\State $\mult(x_2, x_3, 4)$
\EndFor
\end{algorithmic}
It is easy to see that after the above counter program indeed $(x_1, x_2, x_3)$ are equal to
$(4s \cdot 2^n, 16s^2 \cdot 4^n, 0)$ as supposed. Thus it remains to explain how do we realise zero-tests.
We use the controlling-counter technique described in Section~\ref{sec:prelim} and used also in Section~\ref{sec:np}
(and in Section~\ref{sec:tower} later). The counter $x_5$ is the controlling-counter in our counter program and it controls
counters $x_1$, $x_2$ and $x_3$. Recall that in this technique each operation on one of the controlled counters \add{x_i}{a}
is matched by an operation of the controlling-counter \add{x_5}{Na}, where $N$ is the number of zero-tests which will be performed
on the counter $x_i$ in the rest of the run after this operation.
By Lemma~\ref{lem:zero-testing} we know that if value of the controlling-counter $x_5$
is equal to zero at the end of the run then all the zero-tests on controlled counters were correct as well.
Recall also that a bit counterintuitively a zero-test on controlled counter is not reflected in the counter program by any code,
the only effect of a zero-test on some counter $x_i$ is that less zero-tests will be performed on $x_i$ in the future, thus
changes of $x_i$ are reflected now in the controlling-counter in a bit different way ($N$ decreases by one).
Thus the above presented counter program after implementing the zero-tests looks as follows.
\begin{algorithmic}[1]
\State \add{x_1}{4s} \quad \add{x_2}{16s^2}
\For {\, $i$ \, := \, $1$ \, \textbf{to } $n$}\label{l:for}
\Loop \quad \sub{x_1}{1} \quad \add{x_3}{1} \quad \add{x_5}{n+1-i}
\EndLoop
\Loop \quad \add{x_1}{2} \quad \sub{x_3}{1} \quad \sub{x_5}{2}
\EndLoop
\Loop \quad \sub{x_2}{1} \quad \add{x_3}{1} \quad \add{x_5}{n-i}
\EndLoop
\Loop \quad \add{x_2}{4} \quad \sub{x_3}{1} \quad \add{x_5}{2n-2i-1}
\EndLoop
\EndFor
\end{algorithmic}
Let us check carefully that the operations on the controlling-counter $x_5$ are correct.
In line 3 counter $x_1$ awaits for $n+1-i$ zero-tests and counter $x_3$ awaits for $2(n+1-i)$ zero-tests,
so counter $x_5$ should be increased by $(n+1-i) \cdot (-1) +  2(n+1-i) \cdot 1 = n+1-i$.
In line 4 counter $x_1$ awaits for $n-i$ zero-tests and counter $x_3$ awaits for $2(n+1-i)$ zero-tests,
so counter $x_5$ should be increased by $(n-i) \cdot 2 +  2(n+1-i) \cdot (-1) = -2$.
In line 5 counter $x_2$ awaits for $n+1-i$ zero-tests and counter $x_3$ awaits for $2(n+1-i)-1$ zero-tests,
so counter $x_5$ should be increased by $(n+1-i) \cdot (-1) +  (2(n+1-i)-1) \cdot 1 = n-i$.
In line 6 counter $x_1$ awaits for $n-i$ zero-tests and counter $x_3$ awaits for $2(n+1-i)-1$ zero-tests,
so counter $x_5$ should be increased by $(n-i) \cdot 4 +  (2(n+1-i)-1) \cdot (-1) = 2n-2i-1$.
So the above counter program satisfies the conditions of Lemma~\ref{lem:zero-testing}.
Thus by Lemma~\ref{lem:zero-testing} indeed if $x_5 = 0$ at the end of the counter program then
$(x_1, x_2, x_3, x_4)$ have values $(4s \cdot 2^n, 16s^2 \cdot 4^n, 0, 0)$ which finishes the proof.
\end{proof}
\section{\expspace-hardness for 6-VASSes}\label{sec:expspace}
Similarly as in Section~\ref{sec:pspace} using Corollary~\ref{cor:expspace} and the following Lemma~\ref{lem:six-vass}
one can easily derive Theorem~\ref{thm:expspace}. As the argument is totally analogous to the argument in the beginning of Section~\ref{sec:pspace}
and easy to see we do not repeat it here.

\begin{lemma}\label{lem:six-vass}
For each $s, n \in \N$ one can construct in polynomial time
a binary $6$-VASS with distinguished states $q_I, q_F$ such that
for each run from $q_I(0^6)$ to $q_F(x_1, x_2, x_3, x_4, x_5, 0)$
we have $(x_1, x_2, x_3, x_4, x_5) = (6s \cdot 4^{2^n}, 36s^2 \cdot 16^{2^n}, 0, 0, 0)$.
\end{lemma}

The rest of this section focuses on the proof of Lemma~\ref{lem:six-vass}.

\begin{proof}[Proof of Lemma~\ref{lem:six-vass}]
The main idea is quite similar to the proof of Lemma~\ref{lem:five-vass},
namely to set at the beginning $x_1 = 6s$, $x_2 = 36s^2$ and then $2^n$ times multiply the counters $x_1$ and $x_2$ by values
$4$ and $16$, respectively. The multiplications are be realised by the use of the counter $x_3$,
namely we use the macro $\mult$ from the proof of Lemma~\ref{lem:five-vass} with the counter $x_3$ as the second argument.
After these multiplications we also zero-test counter $x_4$ once, the purpose of it will be explained later.
The main difference between the proofs of Lemmas~\ref{lem:five-vass}~and~\ref{lem:six-vass}
is the way how we implement zero-tests. In the proof of Lemma~\ref{lem:five-vass} we used the controlling-counter technique,
but here it is not sufficient and we are forced to use the multiplication triples techniques which is more powerful, but
uses more counters. We are ready to present the demanded counter program,
it roughly speaking looks as follows.
\begin{algorithmic}[1]
\State \add{x_1}{6s} \quad \add{x_2}{36s^2}
\For {\, $i$ \, := \, $1$ \, \textbf{to } $2^n$}\label{l:for}
\State $\mult(x_1, x_3, 4)$
\State $\mult(x_2, x_3, 16)$
\EndFor
\State \ztest($x_4$)
\end{algorithmic}
Here however we cannot expand the macro \textbf{\textup{for}} as it would result in a counter program of exponential size.
We need therefore to show how to implement zero-tests and how to iterate exactly $2^n$ the $\mult$ instructions.
Notice that we need to perform exactly $4 \cdot 2^n$ zero-tests inside the $\mult$ instructions
and then one zero-test on counter $x_4$.
If our counters are $B$-bounded for some $B$ then thanks to Lemma~\ref{lem:triples} in order to simulate
these $4 \cdot 2^n + 1$ zero-tests it is enough to have a triple $(B, 8 \cdot 2^n+2, (8 \cdot 2^n+2) \cdot B)$.
The best bound $B$ such that all $x_1 + x_2 + x_3 + x_4 \leq B$ through the whole run
is $B = 6s \cdot 4^{2^n} + 36s^2 \cdot 16^{2^n}$. At the first moment it looks a bit like a problem, as it seems that
we need one more time to produce triples with doubly-exponential entries.
We perform here however a twist in thinking about triples $(B, C, BC)$,
which was one of the main conceptual contributions of~\cite{DBLP:journals/corr/abs-2105-08551}.
Namely, the triple $(B, C, BC)$ with small $B$ and big $C$ can be both used to implement $C/2$ zero-tests on $B$-bounded counters
and to implement $B/2$ zero-tests on $C$-bounded counters.
In other words: we do not need to compute our big bound $B$, we just need to guess it nondeterministically.
If the guess will be too small then the corresponding run will not be accepting, but it is important that there exists
an appropriate guess for $B$.

Thus our counter program first prepares a triple $(B, C, BC)$ on the counters $x_4$, $x_5$ and $x_6$.
\begin{algorithmic}[1]
\State \add{x_5}{8 \cdot 2^n + 2}
\Loop
\State \add{x_4}{1} \quad \add{x_6}{8 \cdot 2^n + 2}
\EndLoop
\end{algorithmic}
Notice that $n$ is given in unary, as in the reachability problem for bounded three-counter automata $n$ is given in unary.
Thus the above fragment of counter program is of polynomial size, as $8 \cdot 2^n + 2$ can be represented in polynomially many bits,
recall that our $6$-VASS is binary. It is actually the only place where we use the fact that we deal with binary VASSes, but it is an important one.
After this program fragment valuation of $(x_4, x_5, x_6)$ equals $(B, 8 \cdot 2^n + 2, (8 \cdot 2^n + 2)\cdot B)$ for some $B \in \N$.
Then the zero-tests $\ztest(x_i)$ for $i \in \set{1, 2, 3}$ use this triple to perform at most $B$ zero-tests on those counters.

The last part, which remains to be shown is how to afford that the for-loop is fired exactly $2^n$ times.
Recall now that our triple $(B, 8 \cdot 2^n + 2, (8 \cdot 2^n + 2) \cdot B)$ guarantees that
in order to reach at the end of the program some value $(B', 0, 0)$ we need to fire exactly $4 \cdot 2^n + 1$ zero-tests,
which means that the loop needs to be repeated exactly $2^n$ times. Finally observe that after firing these zero-tests
we obtain counter values of $(x_4, x_5, x_6)$ of the form $(B', 0, 0)$, where $B' + 6s \cdot 4^{2^n} + 36 s^2 \cdot 16^{2^n} = B$.
However we actually need $x_4$ to be exactly zero at this point. In order to assure it we apply last zero-test on counter $x_4$.
This zero-test does not differ at all from zero-tests on counters $x_1$, $x_2$ and $x_3$ at all, notice that during the
run we actually keep the invariant $x_1 + x_2 + x_3 + x_4 = B$, so all the counters $x_i$ for $i \in [1,4]$ behave symmetrically
with respect to zero-testing. Notice now that the only guess for $B$ which allows for $x_4 = 0$ at this point
is $B = 6s \cdot 4^{2^n} + 36 s^2 \cdot 16^{2^n}$, in all the other cases the run of our counter program will not reach $x_6 = 0$.
Summarising, the counter program has the following code.
\begin{algorithmic}[1]
\State \add{x_5}{8 \cdot 2^n + 2}
\Loop
\State \add{x_4}{1} \quad \add{x_6}{8 \cdot 2^n + 2}
\EndLoop
\State \add{x_1}{6s} \quad \add{x_2}{36 s^2}
\Loop
\State $\mult(x_1, x_3, 4)$
\State $\mult(x_2, x_3, 16)$
\EndLoop
\State \ztest($x_4$)
\end{algorithmic}
Thus checking whether $x_6 = 0$ at the end of the program indeed assures that the other values are equal
$x_1 = 6s \cdot 4^{2^n}$, $x_2 = 36s^2 \cdot 16^{2^n}$ and $x_3 = x_4 = x_5 = 0$.
\end{proof}

\section{\tower-hardness for 8-VASSes}\label{sec:tower}
Similarly as in Section~\ref{sec:expspace}
due to Lemma~\ref{lem:triples} it is enough to prove the following lemma.

\begin{lemma}\label{lem:towertriple}
For each $n \in \N$ one can construct in polynomial time 
a unary $8$-VASS with distinguished states $q_I, q_F$ such that
for each run from $q_I(0^8)$ to $q_F(x_1, x_2, x_3, x_4, x_5, x_6, x_7, 0)$
we have $x_1 = \tower(n)$, $x_3 = x_2 \cdot \tower(n)$ and $x_4 = x_5 = x_6 = x_7 = 0$.
\end{lemma}

The proof that Lemma~\ref{lem:towertriple} implies Theorem~\ref{thm:tower}
is even simpler than the corresponding one for \expspace-hardness for $6$-VASSes
as we do not need to prove any bound on the number of needed zero-tests;
it follows immediately from Theorem~\ref{thm:bounded-reachability}
and Lemma~\ref{lem:triples}.

Before showing Lemma~\ref{lem:towertriple} we recall first the notion of amplifier and prove
a suitable lemma. Here we define an amplifier in a restrictive setting,
especially adjusted to our application. In particular instead of talking about $f$-amplifier
for $f(k) = 2^k$ we just talk about amplifiers, as we only apply here the notion of amplifiers
to this particular function $f$.

A $7$-VASS $V$ together with its two distinguished states $q_\inp$ and $q_\out$
is an \emph{amplifier} if the following holds:
\begin{itemize}
  \item if $q_\inp(B, C, BC, 0^4) \trans{} q_\out(0^4, B', C', D')$ in $V$ then \newline
  $B'  = 2^n$ and $D' = B' \cdot C'$; and
  \item for each $C' \in \N$ there exists $C \in \N$ such that \newline
  $q_\inp(B, C, BC, 0^4) \trans{} q_\out(0^4, 2^B, C', 2^B \cdot C')$ in $V$.
\end{itemize}

We first show the following lemma.

\begin{lemma}\label{lem:amplifier}
There exists an amplifier.
\end{lemma}

\begin{proof}
We denote counters of the constructed $7$-VASS as $x_i$, for $i \in [1,7]$.
The main idea of the amplifier is that we initialise $(x_5, x_6, x_7)$ as $(1, C', C')$
for some guessed $C' \in \N$ and then multiply $B/8$ times both $x_5$ and $x_7$ by $2^8 = 256$.
The multiplication uses counter $x_4$, similarly as in the proof of Lemma~\ref{lem:six-vass}.
Namely we use the macro $\mult(x, y, c)$ for two counters $x$ and $y$ and a number $c \in \N$ defined as follows.
\begin{algorithmic}[1]
\Loop \quad \sub{x}{1} \quad \add{y}{1}
\EndLoop
\State \ztest($x$)
\Loop \quad \add{x}{c} \quad \sub{y}{1}
\EndLoop
\State \ztest($y$)
\end{algorithmic}
The zero-tests for $x_i$, where $i \in [4,7]$ will use the triple $(B, C, BC)$ on counters $(x_1, x_2, x_3)$
in order to be implemented.
The code of the amplifier is the following.
\begin{algorithmic}[1]
\State \add{x_5}{1}
\Loop \quad \add{x_6}{1} \quad \add{x_7}{1}
\EndLoop
\Loop
\State $\mult(x_5, x_4, 256)$
\State $\mult(x_7, x_4, 256)$
\EndLoop
\Loop \quad \sub{x_2}{1}
\EndLoop
\end{algorithmic}
In each iteration of the loop we fire exactly four zero-tests, twice on counter $x_4$, once on counter $x_5$ and once on counter $x_7$.
Our aim is to have exactly $B / 8$ iterations of the loop.
By Lemma~\ref{lem:triples} using the triple $(B, C, BC)$ guarantees that we can perform exactly $B/2$ zero-tests on $C$-bounded counters.
Notice that here, similarly as in Section~\ref{sec:expspace} we use triples $(B, C, BC)$ in an unusual way: to apply small number ($B / 2$)
of zero-tests on big ($C$-bounded) counters.
As we demand that after finishing the main loop (in lines 2-4) $x_1 = x_3 = 0$ we know that exactly $B/2$ zero-tests were performed,
this implies that exactly $B / 8$ iterations of the main loop were performed, as in each one there are four zero-tests.
Thus each run of our program, which reaches a configuration $(0^4, B', C', D')$ fulfils that $B' = 256^{B/8} = 2^B$ and $D' = B' \cdot C'$,
which means that the program satisfies the first condition of being an amplifier. To show that the second condition of being an amplifier also holds
observe that for each $C'$ is suffices to have $C \geq C' (1 + 2^B)$ and such a $C \in \N$ surely always exists, which finishes the proof
of Lemma~\ref{lem:amplifier}.
\end{proof}

We are now ready to prove Lemma~\ref{lem:towertriple}.

\begin{proof}[Proof of Lemma~\ref{lem:towertriple}]
In the proof for each $n \in \N$ we construct an eight counter program $P_n$ representing the VASS
demanded in the statement of Lemma~\ref{lem:towertriple}.
Very roughly speaking $P_n$ just uses $n$ times the amplifier from Lemma~\ref{lem:amplifier}.
Let $\ampl_i$ denote the code of amplifier from Lemma~\ref{lem:amplifier} with additional
counter $x_8$ with updates depending in the parameter $i$ (to be specified later).
The code of the counter program $P_n$ is roughly speaking the following.
\begin{algorithmic}[1]
\State \add{x_1}{1}
\Loop
\State \add{x_2}{1} \quad \add{x_3}{1}
\EndLoop
\For {\, $i$ \, := \, $1$ \, \textbf{to } $n$}\label{l:for}
\State $\ampl_i$
\State $\ztest(x_1, x_2, x_3, x_4)$
\Loop \quad \sub{x_5}{1} \quad \add{x_1}{1} \EndLoop
\Loop \quad \sub{x_6}{1} \quad \add{x_2}{1} \EndLoop
\Loop \quad \sub{x_7}{1} \quad \add{x_3}{1} \EndLoop
\State $\ztest(x_5, x_6, x_7)$
\EndFor
\end{algorithmic}
Recall here that the $\textbf{for}$-operator is a macro here, so in fact the program $P_n$
has the lines 5-10 repeated $n$ times with different values of $i$ in $\ampl_i$.
Notice also that the eight counter $x_8$ seemingly does not occur in the program.
It is used for implementing the zero-tests, we explain its role in a moment.
Observe however first that if the $\ztest$ procedures correctly zero-test the listed counters then
the program $P_n$ performs what it is supposed to perform.
After lines 1-3 values of seven first counters are $(1, C_0, C_0, 0^4)$ for some arbitrarily guessed $C_0 \in \N$.
It is easy to show by induction on $i$ that after $i$ iterations of the for-loop the counters have values
$(\tower(i), C_i, C_i \cdot \tower(i), 0^4)$ for some arbitrarily guessed $C_i \in \N$.
Indeed, if after the $i-1$ iterations values where $(\tower(i-1), C_{i-1}, C_{i-1} \cdot \tower(i-1))$
then after the amplifier in line 5 of the $i$-th iteration and zero-testing counters $x_i$ for $i \in [1,4]$
by definition of the amplifier counter values are $(0^4, \tower(i), C_i, C_i \cdot \tower(i))$ for some arbitrarily guessed $C_i \in \N$.
Lines 7-10 transfer the triple $(\tower(i), C_i, C_i \cdot \tower(i))$ from counters $(x_5, x_6, x_7)$ to counters $(x_1, x_2, x_3)$
and thus the proof of the induction step is finished.

It remains to show how the zero-tests are implemented. We use here the controlling-counter technique summarised
in Lemma~\ref{lem:zero-testing}. The controlling-counter technique is more useful here than the multiplication triples technique,
because to implement multiplication triple technique we need three additional counters, while to implement a linear number of zero-tests
we need just one additional controlling-counter.
Recall that in the technique of controlling-counter we have an additional controlling-counter (in our case $x_8$) which
starts from zero in the initial configuration. For each other counter (in our case $x_i$ for $i \in [1,7]$) which it controls
each modification of this counter of the form \add{x_i}{a} is matched by a modification of the controlling-counter \add{x_8}{Na},
where $N$ is the number of zero-tests which will be applied to the counter $x_i$ after this modification.
Notice that the commands $\ztest$ in the program $P_n$ are actually not transformed into any real code in this technique,
they only mark a point in the code where the controlling-counter $x_8$ slightly changes its behaviour.
With such a modifications of $x_8$ we are guarantied by Lemma~\ref{lem:zero-testing} that in each run in which
the controlling-counter finishes with value zero all the controlled
counters in all the zero-tested places indeed have value zero. Thus it is enough to add the suitable modifications of the counter $x_8$.
Each of the counters $x_i$ for $i \in [1,7]$ is zero-tested $n$ times in $P_n$, counters $x_1, x_2, x_3, x_4$ in line 6
while counters $x_5, x_6, x_7$ in line 10. Thus we need to add in line 1 operation \add{x_8}{n}
and in line 3 operation \add{x_8}{2n}.
In the $i$-th iteration of the for-loop in line 7 counter $x_5$ awaits for $n-i+1$ zero-tests
while counter $x_1$ awaits for $n-i$ zero-tests. This means that we need to modify counter $x_8$ by $(-1) \cdot (n-i+1) + 1 \cdot (n-i) = -1$.
Similarly in lines 8 and 9 we also need to add the operation \sub{x_8}{1}. Similarly in the program $\ampl_i$ we need
to add for each operation \add{x_i}{a} where $i \in [1,7]$ an operation \add{x_8}{(n-i+1) \cdot a} as each such counter awaits for $(n-i+1)$
zero-tests ($i-1$ of the zero-tests where already performed in the previous $i-1$ iterations of the for-loop). By Lemma~\ref{lem:zero-testing}
we are guarantied that the $\ztest$ operations are correct, thus indeed at the end of $P_n$ we finish the counter
valuation $(0^4, \tower(i), C_i, C_i \cdot \tower(i), 0)$, which finishes the proof of Lemma~\ref{lem:towertriple}.
For clarity we add the code of the counter program $P_n$ below.
\begin{algorithmic}[1]
\State \add{x_1}{1} \quad \add{x_8}{n}
\Loop
\State \add{x_2}{1} \quad \add{x_3}{1} \quad \add{x_8}{2n}
\EndLoop
\For {\, $i$ \, := \, $1$ \, \textbf{to } $n$}\label{l:for}
\State \add{x_5}{1} \quad \add{x_8}{(n-i+1)}
\Loop
\State \add{x_6}{1} \quad \add{x_7}{1} \quad \add{x_8}{(n-i+1)}
\EndLoop
\Loop
\State $\mult(x_5, x_4, 256)$
\State $\mult(x_7, x_4, 256)$
\EndLoop
\Loop \quad \sub{x_2}{1} \quad \sub{x_8}{(n-i+1)} \EndLoop
\Loop \quad \sub{x_5}{1} \quad \add{x_1}{1} \quad \sub{x_8}{1} \EndLoop
\Loop \quad \sub{x_6}{1} \quad \add{x_2}{1} \quad \sub{x_8}{1} \EndLoop
\Loop \quad \sub{x_7}{1} \quad \add{x_3}{1} \quad \sub{x_8}{1} \EndLoop
\EndFor
\end{algorithmic}
where inside the $\mult$ operation in the $i$-th iteration of the for-loop also the operations \add{x_i}{a} for $i \in [1,7]$
are enriched with operations \add{x_8}{(n-i+1) \cdot a}.
\end{proof}

\section{Future research}\label{sec:future}

\paragraph{General remarks}
An obvious future goal is to try to get tight complexity bounds for the reachability problem for fixed
dimensional VASSes. For unary flat VASSes \np-hardness is still open in dimension three
(in dimension two the reachability problem is \nl-complete for unary VASSes~\cite{DBLP:conf/lics/EnglertLT16}).
In general the complexity of the reachability problem in low dimensional VASSes still has a lot of question marks.
For each $d \in [3,7]$ we do not know whether it is elementary or not, moreover for $d \in [3,5]$ for binary encoding
we still cannot exclude that the problem is \pspace-complete, exactly like for $2$-VASSes~\cite{DBLP:conf/lics/BlondinFGHM15}.
In order to exclude \pspace-completeness it would be helpful to come up with some say \expspace-hard or \exptime-hard problem,
which does not involve bounded counter automata but is anyway convenient for a hardness proof;
similarly as \textsc{Subset Sum} is convenient for \np-hardness proof for unary flat $4$-VASSes.

One reason why proving hardness results in low dimensional VASSes is so hard may be because of the use of multiplication
triple technique: we need there three counters to lift our constructions one level higher. Some partial solution to that problem
is the technique of quadratic pairs proposed by us in the paper. It would be interesting to pursue the research in that direction
and try to design some other ways of efficient zero-testing.

\paragraph{Short paths}
A common technique to prove upper complexity bounds on the reachability problem in VASSes
is to show that if there is any reachability path then there is also a short one.
In this way the reachability problem was shown to be in \pspace for binary $2$-VASSes~\cite{DBLP:conf/lics/BlondinFGHM15} 
(reachability path implies exponential length reachability path) and in \nl for unary $2$-VASSes~\cite{DBLP:conf/lics/EnglertLT16}
(reachability path implies polynomial length reachability path).
In particular in order to have hope for \expspace-hardness for $d$-VASSes we need to have an example of $d$-VASS
with the shortest path of at least doubly-exponential length. Similarly for \tower-hardness we need an example of a VASS
with the shortest path being of tower length.
Currently there is a known example of $4$-VASS with shortest path being doubly-exponential~\cite{DBLP:conf/concur/Czerwinski0LLM20}
(Section 5) which means that we may have hope for decreasing the \expspace-hardness from dimension $6$ to $4$ if we happen
to find appropriate techniques. However there are no known examples of $3$-VASSes of shortest reachability path bigger
then exponential and of $7$-VASSes of shortest reachability path bigger then doubly-exponential. Therefore without finding
examples of $3$-VASSes and $7$-VASSes with longer shortest reachability paths we have no hope to prove
\expspace-hardness for $3$-VASSes or \tower-hardness for $7$-VASSes.
This indicates that a search for hard VASS examples may be actually the most needed and potentially fruitful one.

\paragraph{Two-counter automata}
Another way how we can sometimes decrease VASS dimension by one is to use
bounded two-counter automata instead of bounded three-counter automata.
We managed to achieve it for $2^k$-bounded automata in Theorem~\ref{thm:bounded-counters}
and it allowed us to prove Theorem~\ref{thm:pspace} in dimension $5$ instead of $6$.
However this technique does not seem to extend immediately to higher bounds, for example to $2^{2^k}$-bounded automata.
To our best knowledge the following statement is open, but we conjecture it to be true.
\begin{conjecture}\label{conj:twocounters}
The reachability problem for $f$-bounded two-counter automata (with unary updates) is \expspace-complete for $f(k) = 2^{2^k}$.
\end{conjecture}
This conjecture would not immediately give \expspace-hardness for binary $5$-VASSes as
generating the pair $(6s \cdot 4^{2^n}, 36s^2 \cdot 16^{2^n})$ in the proof of Theorem~\ref{thm:expspace}
currently needs six counters, but would be some step towards it and an interesting result in itself.

\paragraph*{Acknowledgements}
We thank Sławomir Lasota for letting us to present his proof of the first item in Theorem~\ref{thm:bounded-counters}.

\bibliographystyle{plain}
\bibliography{citat}

\appendix
\section{Missing proof}
We recall the statement of Theorem~\ref{thm:bounded-counters}.

\vskip 0.3cm

\noindent
\textbf{Theorem~\ref{thm:bounded-counters}.} The reachability problem for $f$-bounded $d$-counter automata is
\begin{enumerate}
  \item \pspace-hard for $f(n) = 2^n$ and $d = 2$
  \item \expspace-hard for $f(n) = 2^{2^n}$ and $d = 3$.
\end{enumerate}

\vskip 0.3cm

\begin{proof}[Proof of Theorem~\ref{thm:bounded-counters}]
The proof of (1) is due to~\cite{SlawekPspace}, while the proof of (2) is a small modification of the classical
proof from~\cite{FischerMR68}.

For the (sketch of the) proof of (1) we reduce from the reachability problem for linear bounded automata,
which is a classical problem known to be \pspace-hard~\cite{DBLP:books/aw/HopcroftU79}.
We can formulate the problem as follows: we are given a Turing machine (TM) with $B$ tape letters: $0, 1, \ldots, B-1$ and tape is of size $n$.
The problem is to decide whether there is a run of this TM starting in a distinguished state $q_I$ with whole tape covered with symbol $0$
and finishing in another distinguished state $q_F$ with also the whole tape covered with symbol $0$ such that this run only uses these
$n$ cells of the tape. The idea is to encode the configuration of the tape by one counter. Concretely speaking the configuration of
the tape with letter $a_i \in [0,B-1]$ on cell $i \in [1,n]$ is encoded as $p_1^{a_1} \cdot p_2^{a_2} \cdot \ldots \cdot p_n^{a_n}$,
where $p_1, \ldots, p_n$ are the $n$ smallest prime numbers. The Prime Number Theorem (PNT) (see for example~\cite{PNT}) says
that $\Pi(n) \approx n / \log(n)$, where $\Pi(n)$ is the number of primes in the interval $[1,n]$. 
More concretely speaking PNT can be formulated as $\lim_{n \to \infty} \frac{\pi(n)}{n / \log(n)} = 1$.
An easy consequence is that there exists a universal constant $C \in \N$ such that for each $n \in \N$ the $n$-th prime number $p_n$
satisfies $p_n \leq C \cdot n\log(n)$.
Thus we have that
\[
p_1^{a_1} \cdot p_2^{a_2} \cdot \ldots \cdot p_n^{a_n} \leq (C \cdot n\log(n))^{n(B-1)},
\]
which can be bounded from above by $2^{n^2}$ for arbitrarily big $n$, so it is bounded by $C' \cdot 2^{n^2}$ for some $C' \in \N$.
We show now that for an input TM one can construct a two-counter automaton $\A$ and a number $n' = C' \cdot n^2$
such that $\A$ is $2^{n'}$-bounded iff the input TM has an accepting run.
The state of TM and position of its head is kept in the state of $A$. We only need now to show how to simulate transitions of TM by a two-counter
automaton. Assume that TM has its head over the $i$-th cell and it fires a transition if a letter $a \in [0,B-1]$ is written in the $i$-th cell
writing there instead a letter $b \in [0,B-1]$. Then the corresponding automaton $\A$ needs to divide the number representing tape of TM
by $p_i^a$, check that it is no more divisible by $p_i$ (thus checking that indeed latter $a$ was in $i$-th cell) and then multiply it by $p_i^b$.
All these operations can be easily performed on a counter in two-counter automaton; the second counter is used as an auxiliary counter
for multiplications and divisions via $p_i$. Of course during this transition also the state of $\A$ needs to be modified in an appropriate way.
Thus reachability in the TM can be indeed reduced to the existence of the run in $\A$. Notice now that the starting value of the counter
representing tape is equal to $p_1^0 \cdot \ldots \cdot p_n^0 = 1$ and as well the final value.
So we have reduced existence of the accepting run in TM
to the existence of the run from $p_I(1,0)$ to $p_F(1,0)$ where $p_I, p_F$ are the initial and final states of $\A$, respectively.
The bounded counter automata start from $0^d$ valuation and finish in $0^d$ as well, so to finish the argument we just need to add auxiliary
initial state $p'_I$ with a transition to $p_I$ with the effect $(1,0)$ and auxiliary final state $p'_F$ with a transition from $p_F$
with the effect $(-1,0)$. Then existence of an accepting run in the input TM reduces to existence of the $2^{n'}$-bounded run from $p'_I(0,0)$
to $p'_F(0,0)$ in $\A$, which finishes the proof.

We prove (2) following the classical lines of the proof from~\cite{FischerMR68}. We only briefly sketch the solution.
We reduce from the reachability problem for Turing machines with exponential memory (assume that size of this memory is $2^n$).
Let us assume that there are $B$ tape symbols, we encode them as numbers $1, \ldots, B$
(this is a small modification wrt. to the classical encoding by $0, \ldots, B-1$).
We encode the tape by two numbers encoding: the part of the tape to the left of the head and to the right of the head, respectively.
Let us assume that the letter below the head is kept in the state of the three-counter automaton.
Tape contents are encoded classically as numbers in base $B+1$ with the least significant digit being the one closest to the head.
The simulation of the transition is also performed classically, namely if the head is moved left then the number encoding
the left part of the tape should be divided by $B$ and number encoded the right part of the tape should be multiplied by $B$
and the digit below the previous location of the head should be added to it.
These operations can be easily implemented by the use of the third counter.
Note now that as the digits encoding the tape belong to $[1,B]$ we know when the head is on the left or right border of the memory
as then the corresponding counter encoding tape content is equal to zero. 
Thus we can disallow in such cases moving the head further to the left or to the right, respectively.
Therefore the maximal sum of values of the counters is bounded from above by $B^{2^n} \leq 2^{2^{Bn}}$.
In such a way we reduced the reachability problem of TM with exponential memory to the reachability problem of $2^{2^{Bn}}$-bounded
three-counter automata. However we have slight technical complication here: the starting and finishing configurations do not
have counter values equal to zero (or one as in the proof of (1)), but pretty big values. If at the beginning of the run head is on the leftmost cell
then the counter encoding the part of the tape to the right of the head has value $1 + (B+1) + \ldots + (B+1)^{2^n-1}$.
Similar situation occurs at the end of the run. Notice however that this problem can be easily solved.
We can append to our automaton a pre-computation
of this big number (which can be easily performed using three zero-tested counters) and similarly append a post-computation which
decreases the corresponding counter by an appropriate value. This finishes the sketch of the proof of point (2).
\end{proof}

\end{document}